\documentclass[11pt,letterpaper]{article}
\usepackage{fullpage}
\usepackage{amsmath,amsthm,amssymb,color}
\usepackage[utf8]{inputenc}
\usepackage[T1]{fontenc}
\usepackage{graphicx}

\usepackage[ruled,vlined,linesnumbered,noalgohanging,nofillcomment]{algorithm2e}
\SetAlgoHangIndent{0pt}
\DontPrintSemicolon
	
\newtheorem{theorem}{Theorem}[section]
\newtheorem{lemma}[theorem]{Lemma}

\newtheorem{corollary}[theorem]{Corollary}

\newtheorem{fact}[theorem]{Fact}

\DeclareMathOperator{\poly}{poly}
\DeclareMathOperator{\polylog}{polylog}

\newcommand{\eqdef}{\stackrel{\text{\tiny\rm def}}{=}}

\newcommand{\SET}[1]{\{#1\}}
\newcommand{\MAX}[1]{\max\{#1\}}
\newcommand{\MIN}[1]{\min\{#1\}}

\newcommand{\eps}{\epsilon}

\newcommand{\calA}{{\mathcal A}}
\newcommand{\calI}{{\mathcal I}}
\newcommand{\calG}{{\mathcal G}}

\newcommand{\tO}{{\tilde O}}

\newsavebox{\algDescBox}
\savebox{\algDescBox}{\phantom{\bf Algorithm 5:}}

\newcommand{\algdesc}[1]{\newline \leavevmode\usebox{\algDescBox}\,\,#1}

\newcommand{\globalAlg}{\mbox{\tt GlobalPeeling}}
\newcommand{\RCAlg}{\mbox{\tt RoundCompression}}
\newcommand{\arbMIS}{\mbox{\tt ArboricityMIS}}
\newcommand{\MPCMatch}{\mbox{\tt MatchMPC}}
\newcommand{\AlgFewPhases}{\mbox{\tt LocalPeeling}}

\newcommand{\LOCAL}{{\rm LOCAL}}

\title{Round Compression for Parallel Graph Algorithms\\in Strongly Sublinear Space}

\author{Krzysztof Onak\vspace{3pt}\\\includegraphics[height=10pt]{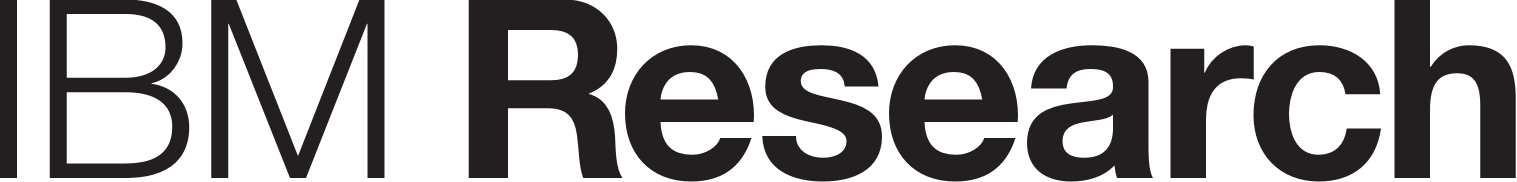}}

\date{July 2018}

\begin{document}

\maketitle

\begin{abstract}
The Massive Parallel Computation (MPC) model is a theoretical framework for popular parallel and distributed platforms such as MapReduce, Hadoop, or Spark. We consider the task of computing a large matching or small vertex cover in this model when the space per machine is $n^\delta$ for $\delta \in (0,1)$, where $n$ is the number of vertices in the input graph. A direct simulation of classic PRAM and distributed algorithms from the 1980s results in algorithms that require at least a logarithmic number of MPC rounds. We give the first algorithm that breaks this logarithmic barrier and runs in $\tO(\sqrt{\log n})$ rounds, as long as the total space is at least slightly superlinear in the number of vertices.

The result is obtained by repeatedly compressing several rounds of a natural peeling algorithm to a logarithmically smaller number of MPC rounds. Each time we show that it suffices to consider a low--degree subgraph, in which local neighborhoods can be explored with exponential speedup. Our techniques are relatively simple and can also be used to accelerate the simulation of distributed algorithms for bounded--degree graphs and finding a maximal independent set in bounded--arboricity graphs.
\end{abstract}

\section{Introduction}

The \emph{Massive Parallel Computation} (MPC) model is a neat framework introduced by Karloff, Suri, and Vassilvitskii~\cite{MPC} to describe efficient computation in MapReduce~\cite{dg04,dg08} and applies to other modern massive distributed computation platforms that are widely successful in practice such as Hadoop~\cite{White:2012}, Dryad~\cite{Isard:2007}, or Spark~\cite{ZahariaCFSS10}. It is often possible to simulate classic PRAM or distributed algorithms from the LOCAL model directly in MPC, using roughly the same number of parallel computation rounds (assuming that they are not too extravagant or inefficient). See the works by Karloff et al.~\cite{MPC} and Goodrich, Sitchinava, and Zhang~\cite{GSZ} for a more detailed discussion of this topic. It is, however, an appealing challenge to design algorithms that solve a given combinatorial problem in much fewer MPC rounds by taking advantage of the different design of the model, which allows for arbitrary local computation on a small fraction of data and reshuffling information globally in each computation round. 

An important parameter of the MPC model is the amount $S$ of space assigned to a single machine. Throughout this paper, which focuses on graph algorithms, we assume that $n$ is the number of vertices in the input graph. Previous research has mostly focused on the regime of $S = n^{1+\Theta(1)}$~\cite{MPC,LattanziMSV11,AhnG15}, or more recently, $S=\tilde\Theta(n)$~\cite{round_compression,Assadi_VC,loglog_1,loglog_2}. These lines of work have resulted in algorithms that run in $O(1)$ or $O(\log \log n)$ MPC rounds, which is significantly faster than $\Omega(\log n)$ required by the best algorithms in the aforementioned classic models of computation.

Only very recently a number of works have considered the space regime $S=O(n^\delta)$ for a fixed $\delta \in (0,1)$ and gave algorithms with strongly sublogarithmic numbers of rounds. To the best of our knowledge, all of them require an input from a restricted class of graphs. For instance, Brandt, Fischer, and Uitto~\cite{BFU} give an algorithm for computing a maximal independent set in trees. Andoni, Stein, Song, Wang, and Zhong~\cite{connect_1} and Assadi, Sun, and Weinstein~\cite{connect_2} give connectivity algorithms for graphs that have a limited diameter or are well--connected. In this work, we give an $\tilde O(\sqrt{\log n})$--round MPC algorithm for approximating maximum matching and vertex cover in \emph{arbitrary graphs}. Until now, it was only known how to achieve a $2$--approximation in $\Theta(\log n)$ rounds by simulating classic maximal independent set and maximal matching algorithms of Luby~\cite{Luby86}, Alon, Babai, and Itai~\cite{AlonBI86}, and Israeli and Itai~\cite{II86}.

Let us briefly mention the importance of the $S=O(n^{\delta})$ regime. Many big graphs in practice, such as social networks or networks of financial transactions, are sparse with the number of edges linear in the number of vertices. Given their size, they may not fit onto one or a small number of machines. Therefore, it may be very useful to distribute both data and processing across a larger cluster of machines. We note that the recent line of work on the near--linear regime~\cite{round_compression,Assadi_VC,loglog_1,loglog_2} allowed for a slightly sublinear amount of space per machine such as $S = n / \polylog(n)$. The $S=O(n^{\delta})$ regime enables, however, a significantly wider range of potential applications.

\paragraph{Recent developments.} Concurrently to this work, a few works have considered MPC algorithms for graph problems in the same space regime~\cite{arboricity-1,arboricity-2,GU}. In particular, the paper by Ghaffari and Uitto~\cite{GU} considers the same problems and uses similar techniques. We are still investigating the full relationship of these works to ours.

\subsection{Massive Parallel Computation}

In the Massive Parallel Computation model~\cite{MPC}, there are $M$ machines and each of them has $S$ words of space. Focusing on the scenario considered in this paper, the input is a set of $m$ edges and initially, each machine receives a fair share of roughly $m/M$ of them. The computation proceeds in rounds. During each round, machines first process their local data without communicating between each other. Then machines create and send messages addressed to each other. Each message is sent to only a single machine specified by the sender. An important constraint is that all messages sent and received by a single machine have to fit into the machine's local space of size $S$. The messages can be processed by recipients in the next round. At the end of the computation, machines can collectively output the solution. Each machine's output has to fit into that machine's local space and therefore each machine can output at most $S$ words.

In order for the computation to be possible, the total space $M \cdot S$ has to be at least linear in the input size---which is $m$ in our case---and preferably not significantly larger. In this work, we generally allow the total space to be of order $n^{1+o(1)} + O(m)$. 
The original definition allowed for nearly--quadratic total space---$N^{1-\eps}$ machines with $N^{1-\eps}$ space each, where $N$ is the input size and $\eps$ is a small fixed constant---but later works~\cite{BeameKS17,ANOY} suggested focusing on near--linear total space, which may be significantly more practical in the big data setting.

We also mention that if $S \ge n^\delta$ for a fixed constant $\delta \in (0,1)$, then a number very useful primitives can be simulated in $O(1)$ MPC rounds such as sorting, prefix--sum computation, etc.~\cite{GSZ}. We (and most other works in the area) heavily rely on them.

\subsection{Maximum Matching and Minimum Vertex Cover}

A set of edges that share no vertices is called a \emph{matching}. In the \emph{maximum matching} problem, the goal is to find a matching of maximum cardinality. If $M_\star$ is a maximum matching, then a matching $M$ is a $p$--approximation, for $p \ge 1$, if $|M_\star|\le p |M|$.
A set $C$ of vertices is a \emph{vertex cover} if for every edge $\{u,v\}$ in the graph at least one of $u$ and $v$ belongs to $C$. In the \emph{(minimum) vertex cover} problem, the goal is to find a vertex cover of minimum cardinality. Let $C_\star$ be a minimum vertex cover. We say that a vertex cover $C$ is a $p$--approximation, where $p \ge 1$, if $|C| \le p|C_\star|$.

If $S = n^{1+\Omega(1)}$, then it is known that one can compute a maximal matching in $O(1)$ MPC rounds~\cite{LattanziMSV11} and the approximation factor can be improved to $(1+\eps)$ in $O(1/\eps)$ rounds~\cite{AhnG15}. If $S = O(n)$, the an $O(2+\eps)$--approximation to maximum matching and vertex cover can be computed in $O(\log(1/\eps) \cdot \log \log n)$ rounds~\cite{round_compression,loglog_1,loglog_2}. For maximum matching, the approximation factor can be improved to $(1+\eps)$ at the cost of an additional factor of $(1/\eps)^{O(1/\eps)}$ in the number of rounds.

The simplest version of our result can be stated as follows.

\begin{theorem}\label{thm:MM_simple}Let $\delta \in (0,1)$ be a fixed constant. There is an MPC algorithm that for an $n$--vertex graph, computes a constant--factor approximation to both maximum matching and minimum vertex cover, runs in $\tilde O(\sqrt{\log n})$ rounds, and uses $O(n^\delta)$ space per machine and $O(m) + n^{1+o(1)}$ total space.
\end{theorem}

See Theorem~\ref{thm:MM_VC} for the full statement that allows for balancing the number of rounds and the magnitude of the extra factor of $n^{o(1)}$ in the total space bound.
We also discuss further ramifications of the result in Section~\ref{sec:mm_extensions}, including obtaining better approximation factors and the weighted matching problem.

\subsection{Maximal Independent Set in Bounded--Arboricity Graphs}

A subset $U$ of vertices of a graph is an \emph{independent set} if for any edge $\{u,v\}$ at most one of $u$ and $v$ belongs to $U$. Additionally, an independent set is a \emph{maximal independent set} if it cannot be extended, i.e., adding any vertex would result in a set that is not independent. In the maximal independent set problem, the goal is to compute any maximal independent set in the input graph.

If $S=n^{1+\Theta(1)}$, then a maximal independent set can be computed in $O(1)$ MPC rounds~\cite{HarveyLL18}. It is known that if $S = O(n)$, then a maximal independent set in an arbitrary graph can be found in $O(\log \log n)$ MPC rounds \cite{loglog_2,MIS_in_CQ}. For $S = O(n^\delta)$, the best known algorithm remains the direct simulation of the classic algorithms of Luby~\cite{Luby86} or Alon, Babai, and Itai~\cite{AlonBI86} in $O(\log n)$ rounds. Recently, Brandt, Fischer, and Uitto~\cite{BFU} showed an algorithm that in this memory regime finds a maximal independent set in trees in $O((\log \log n)^3)$ rounds. We consider a broader class of bounded--arboricity graphs (i.e., graphs that can be decomposed into a small number of forests; arboricity $\alpha$ means that the graph decomposes into at most $\alpha$ forests), but our number of rounds is significantly higher. The simplest version of our result can be stated as follows.

\begin{theorem}\label{thm:arbMIS_simple}
Let $\calG$ be a class of graphs of arboricity $O(1)$ and let $\delta \in (0,1)$ be a fixed constant. There is an MPC algorithm that for an $n$--vertex graph from $\calG$, computes a maximal independent set with probability $1 - O\left(\frac{\log n}{n}\right)$, runs in $\tO(\sqrt{\log n})$ rounds, and uses $O(n^\delta)$ space per machine and $n^{1+o(1)}$ total space.
\end{theorem}

See Theorem~\ref{thm:arbMIS} for the full statement that allows for balancing the number of rounds and the magnitude of the extra factor of $n^{o(1)}$ in the total space bound.

\subsection{Our Techniques}

Our techniques are relatively simple. For both main problems that we consider in this paper, we repeatedly extract a low--degree subgraph and solve a problem directly on it by exponentially accelerating the number of rounds that it needs. This comes at the cost of slightly increasing the total space by a factor of $n^{o(1)}$ beyond $O(n)$, because we double the radius of the neighborhood known to each vertex in each round. We note that recent works on connectivity in strongly sublinear space~\cite{connect_1,connect_2} also consider the problem of growing local neighborhoods. They have use more sophisticated methods in order to deal with unbounded vertex degrees and we use a more brute--force approach in which we completely collect local neighborhoods up to a specific radius.

For maximum matching and vertex cover, we extract a low--degree subgraph that corresponds to a superconstant number of iterations (out of a total $O(\log n)$) of a distributed algorithm of Onak and Rubinfeld~\cite{OR}. A similar approach has been used in other memory regimes~\cite{AssadiK17,round_compression,Assadi_VC}.

For the maximal independent set problem on bounded--arboricity graphs, we decrease the total number of vertices each time by extracting the graph on low--degree vertices and simulating the algorithm of Ghaffari~\cite{G-MIS}.

We also note that the idea of locally simulating a distributed algorithms in other models of computation was used by Parnas and Ron~\cite{PR} in order to provide sublinear--time estimation algorithms.

\subsection{Comparison to Connectivity} Our work exhibits an interesting relationship between the complexity of approximating the maximum matching (or vertex cover) and computing connectivity. If the space per machine is $\tilde O(n)$, we know that one can compute connected components in $O(1)$ MPC rounds~\cite{connect_lin_1,connect_lin_2}, while the best algorithms for computing the exact matching require $O(\log \log n)$ MPC rounds~\cite{loglog_1,loglog_2}. More specifically, approximating maximum matching seems more difficult in this regime.

On the other hand, if the space per machine becomes strongly sublinear in $n$, i.e., at most $n^{\delta}$ for $\delta \in (0,1)$, then our work shows that a good approximation to maximum matching can be computed in $\tilde O(\sqrt{\log n})$ rounds. At the same time, it seems unlikely that $(\log n)^{1-\Omega(1)}$ rounds suffice to even determine the connectivity of the input graph
(see, for instance, the recent works on connectivity in this space regime for a more--in--depth discussion~\cite{connect_1,connect_2}). Interestingly, in this regime, connectivity becomes a seemingly more difficult problem.

\section{Round Compression for Low--Degree Graphs}\label{sec:rc}

In this section, we introduce a simple technique for simulating a small number of rounds of a~distributed algorithm in MPC. The main idea is that for a deterministic distributed algorithm $\calA$, the output of each vertex is a function of solely the neighborhood of radius equal to the number of simulated rounds. This neighborhood can be collected in a number of MPC rounds that is logarithmic in the number of simulated rounds. We achieve this result by doubling the radius of the neighborhood known to each vertex in a constant number of rounds.

\begin{lemma}\label{lem:rc}
Let $\delta \in (0,1)$ be a fixed constant. Let $l_V$ and $l_E$ be the length of labels assigned to vertices and edges, respectively, expressed in words. Let $\calA$ be a deterministic distributed algorithm in the \LOCAL{} model that runs in $t \ge 2$ rounds on a graph of degree bounded by $d \ge 2$ and does not use more space than $s_\calA$ at any time to compute the output at each vertex.
Let $s_\star \eqdef d^{t}\left(l_V  + d(1+l_E) + s_\calA\right)$.

There is an MPC algorithm that computes the output of $\calA$ at each vertex, runs in $O(\log t)$ rounds, and uses $O(\MAX{n^\delta,s_\star})$ space per machine and $O(ns_\star)$ total space.
\end{lemma}

\begin{proof}
For any vertex $v$ in input graph $G = (V,E)$, let $N_i(v)$ denote the subgraph of $G$ consisting of all vertices at distance at most $i$ from $v$ and all their incident edges, including all associated labels. Note that this definition is slightly unusual as the subgraph may contain an edge but not one of its endpoints. For instance, if $v$ has a neighbor $u$, which has a neighbor $w$ that is not a neighbor of $u$, then $N_i(v)$ contains $u$ and the edge $(u,w)$, both with their labels, but it does not contain $w$ and its label. The reason for this definition is that $N_i(v)$ captures exactly the information on which the output at $v$ can depend for any deterministic distributed algorithm that runs in $i$ round. 

Suppose now that the maximum degree in $G$ is bounded by $d \ge 2$. Consider any vertex $v \in V$. Let $n_i$ be the maximum number of vertices at distance at most $i$, for any $i \ge 0$. Since $d \ge 2$, we have $n_i = O(d^i)$. Therefore, the number of edges in $N_i(v)$ is at most $dn_i = O(d^{i+1})$. It is easy to construct a description of $N_i(v)$ of size at most $s_i = O(d^i(1+l_V) + d^{i+1}(1+l_V)) = O(d^i(l_V+d(1+l_E)))$ words.

\begin{algorithm}[ht]
  \caption{$\RCAlg(G,\calA,t)$\algdesc{Round compression for low--degree graphs}}\label{alg:rc}
  \KwIn{\\
  \qquad$\bullet$ graph $G = (V,E)$ with degree of all vertices at most $d$\\
  \qquad$\bullet$ deterministic distributed algorithm that runs in $t$ rounds\\
  \qquad$\bullet$ the number of rounds $t$ for which $\calA$ will be simulated}
  \KwOut{the result of running $\mathcal A$ on $G$}

\BlankLine

  $r \leftarrow 0$
  
  Distribute vertices $v \in V$ evenly among the machines and send $N_0(v)$ to each of them
  
  \While{$r < t$}{
  
      $r' \leftarrow \min\{r,t-r-1\}$
  
      \ForEach{\rm $v,w \in V$ such that distance of $w$ from $v$ is $r+1$}{Send $N_{r'}(w)$ to the machine assigned to $v$}

      \ForEach{$v \in V$}{Combine $N_r(v)$ with all received $N_{r'}(w)$ to obtain $N_{r+r'+1}(v)$}
      
      $r \leftarrow r + r'+ 1$
  }

  \ForEach{$v \in V$}{Simulate $\calA$ on $N_t(v)$ to compute the output of $\calA$ at $v$}
  
\end{algorithm}

We present the MPC algorithm for computing the output of $\calA$ as Algorithm~\ref{alg:rc}. This algorithm distributes all $n$ vertices evenly between the machines. This can be implemented by sorting unique vertex identifiers in $O(1)$ MPC rounds. Then, in $O(\log t)$ rounds, we collect $N_t(v)$ for each vertex $v \in V$. For each increasing radius parameter $r$, the total length in words of all messages that are sent and received for each vertex $v$ is bounded by $d^{r+1} \cdot s_{r'} = d^{r+1} \cdot O\left(d^{r'}(l_V+d(1+l_E))\right) = O\left(d^{t}(l_V+d(1+l_E))\right) = O(s_\star)$. Hence $O(\MAX{s_\star,n^\delta})$ space per machine and $O(ns_\star)$ total space suffices to conduct this step. Then for each $v\in V$, the machine assigned to $v$ simulates all the $t$ rounds of $\calA$ on $N_t(v)$ to compute the output $\calA$ would compute at $v$. Since the number of vertices in each $N_t(v)$ is bounded by $n_i$, the additional space needed for simulating $\calA$ is also bounded by $O(n_t s_\calA) = O(d^t s_\calA) = O(s_\star)$. Hence there is enough space to run this simulation as well.
\end{proof}

\paragraph{Sample application.}
{\AA}strand et al.~\cite{AstrandFPRSU09} give a distributed algorithm in the LOCAL model that computes a 2--approximation to vertex cover in graphs of degree at most $\Delta$. It is deterministic and runs in $(\Delta+1)^2$ rounds. Furthermore, it uses at most $\poly(\Delta)$ space per vertex throughout the execution. By applying Lemma~\ref{lem:rc}, we obtain an MPC algorithm that for graphs of degree bounded by $\Delta$, computes a 2--approximation to vertex cover in $O(\log \Delta)$ rounds, and uses $\MAX{O(n^\delta),2^{\tO(\Delta^2)}}$ space per machine and $n\cdot 2^{\tO(\Delta^2)}$ total space, where $\delta > 0$ is a fixed constant.

If the extra factor of $2^{\tO(\Delta^2)}$ is too high, we can reduce the total space usage at the cost of a larger number of rounds. For instance, we can partition the execution of the algorithm into $\Delta+1$ phases, each consisting of $\Delta+1$ rounds. By applying round compression repeatedly to each of the phases (and saving the intermediate state of the algorithm at each vertex), we obtain an MPC algorithm that computes the same output in $O(\Delta \log \Delta)$ rounds, and uses $\MAX{O(n^\delta),2^{\tO(\Delta)}}$ space per machine and $n\cdot 2^{\tO(\Delta)}$ total space.

\paragraph{Note on randomness.}
Note that we assume that $\calA$ in Lemma~\ref{lem:rc} is deterministic. This is necessary so that the behavior of the algorithm on every node is fixed and any machine can compute the outcome of computation at any vertex. As we see later, this limitation can, however, be circumvented. This is achieved by assigning the randomness needed by the algorithm to either vertices or edges in the form of labels. A randomized algorithm is then transformed into a deterministic algorithm that depends on the labels (and uses them for its source of ``randomness'').

\section{Algorithm for Maximum Matching and Vertex Cover}

\subsection{Review of the Peeling Algorithm}

In this section, we review the peeling algorithm for maximum matching and 
vertex cover that was introduced by Onak and Rubinfeld~\cite{OR}. It was 
inspired by the $O(\log n)$--approximation algorithm for vertex cover of Parnas and Ron~\cite{PR}.
Both these algorithms are (or can be seen as) distributed algorithms in the LOCAL model.

\begin{algorithm}[ht]
  \caption{$\globalAlg(G,d)$\algdesc{A peeling algorithm that we want to simulate}}\label{alg:peeling_central}
  \KwIn{graph $G = (V,E)$ with degree of all vertices at most $d$}
  \KwOut{matching and vertex cover in $G$}

\BlankLine
  
  $\Delta \leftarrow d$, $U \leftarrow V$, $M \leftarrow \emptyset$, $C \leftarrow \emptyset$
  
  \While{$\Delta \ge 1$\label{line:peeling:main_loop}}{
  
      \tcc{Invariant:\ maximum degree in $G[U]$ at most $\Delta$}

      $\Delta \leftarrow \Delta/2$

      Let $H \subseteq U$ be the set of vertices of degree at least $\Delta$ in $G[U]$\label{line:peeling:heavy}
      
      \tcc{We call vertices in $H$ heavy\hbox{.}}

      \ForEach{$v\in H$\label{line:peeling:select_f}}{$f(v) \leftarrow \mbox{neighbor of $v$ in $G[U]$ selected uniformly independently at random}$ 
      
      \tcc{We call $f(v)$ a friend of $v$\mbox{.}}}

      $F \leftarrow \{f(v): v \in H\}$\label{line:peeling:friends}
      
      Color each vertex in $H \cup F$ either blue or red independently at random\label{line:peeling:color}
      
      $M \leftarrow M \cup \left\{(v,f(v)) \in H \times F : \mbox{$v$ is blue} \land \mbox{$f(v)$ is red} \land \forall_{\mbox{\scriptsize blue\ }w \in H\setminus\{v\}}\mbox{$f(w) \ne f(v)$}\right\}$\label{line:peeling:match}
      
      $C \leftarrow C \cup H \cup F$\label{line:peeling:vc}
      
      $U \leftarrow U \setminus \left(H \cup F\right)$\label{line:peeling:reduce_graph}
  
  }
  
  \Return{$(M,C)$}
\end{algorithm}

We present the pseudocode of the algorithm as Algorithm~\ref{alg:peeling_central}. We now briefly discuss how it works and present the intuition behind it. 
The algorithm takes as input a graph $G$ and upper bound $d$ on the maximum vertex degree in $G$. It proceeds in a number of phases that is logarithmic in $d$. Throughout the execution, it maintains a threshold $\Delta$, which initially equals $d$. In the process, the algorithm keeps removing vertices from the original graph---$G[U]$ denotes the current graph---while  adding edges to a matching and vertices to a vertex cover. At the beginning of each phase, the threshold is halved, i.e., after consecutive phases is becomes $d/2$, $d/4$, $d/8$, and so on. An important maintained invariant is that the maximum degree in the remaining graph is bounded by $\Delta$ before and after each phase.

We now discuss the main loop of the algorithm in Lines~\ref{line:peeling:main_loop}--\ref{line:peeling:reduce_graph}.
In each phase, after halving $\Delta$, the algorithm selects $H$, the set of vertices of degree at least $\Delta$ (see Line~\ref{line:peeling:heavy}). We refer to them as \emph{heavy}. Then, for each $v\in H$, it selects uniformly and independently at random a neighbor $f(v)$, which we call a \emph{friend} of $v$, and it also defines $F$ to be the set of all friends (see Lines~\ref{line:peeling:select_f}--\ref{line:peeling:friends}). What happens in Lines~\ref{line:peeling:color} and~\ref{line:peeling:match} can be replaced by any algorithm that finds a large matching between $H$ and $F$. The procedure specified here is easier to analyze for our purposes. First, it turns the graph into a bipartite graph by randomly coloring each vertex either blue or red. Then for each heavy vertex $v$, it adds $(v,f(v))$ to the matching if $v$ is blue and $f(v)$ is red, and no other blue heavy vertex $w$ claims $f(v)$ to be its friend $f(w)$. Finally, in Lines~\ref{line:peeling:vc} and~\ref{line:peeling:reduce_graph}, all heavy vertices and their friends are added to the vertex cover and removed from the remaining graph.

It is easy to verify that the constructed sets $M$ and $C$ are a proper matching and vertex cover. This algorithm produces a constant--factor approximation to both maximum matching and vertex cover with constant probability. In each phase, one can show that the expected size of the set of edges added to $M$ is $\Omega(|H|)$. Consider any heavy vertex $v$. Since the degree of all vertices in $G[U]$ is bounded by $2\Delta$ in Line~\ref{line:peeling:match}, the probability that another blue heavy vertex requests $f(v)$ to be its friend is at most a constant bounded away from 1. One can then show that this relationship holds as well for the final $M$ and $C$. Therefore, by Markov's inequality, with constant probability $M$ cannot lose too much compared to $C$. If that is the case, $C$ and $M$ are within a constant factor and are constant factor approximations, because any vertex cover has to be of size at least $|M|$ to cover edges of $M$.

An important and useful property of the algorithm is that it can be simulated in the MPC model exactly, using $O(1)$ MPC rounds to simulate each iteration of the main loop. Essentially, since sorting and prefix sums can be computed in $O(1)$ MPC rounds, it is possible to compute exact degrees vertices and then select heavy vertices, random neighbors, and isolates edges in $O(1)$ MPC rounds as well with all sets represented as lists, which may not fit onto a single machine and span several machines. (The topic of exact simulation of a version of this algorithm is also discussed by Czumaj et al.~\cite{round_compression} in the full version of their paper, which can be found on arXiv.) We state this as the following fact.

\begin{fact}\label{fact:sim_global_peel}
Let $\delta \in (0,1)$ be a fixed constant. Let $G$ be an $n$--vertex graph with $m$ edges and maximum degree at most $d > 1$. \globalAlg{} can be simulated in $O(\log d)$ MPC rounds and $O(n^\delta)$ space per machine and $O(n+m)$ total space.
\end{fact}

\subsection{The MPC Algorithm}

In our more efficient implementation of \globalAlg{} in MPC, we do not maintain exact degrees of vertices in the remaining graph, but instead use approximation. We use sampling and due to standard concentration bounds, we know know that the degrees of heavy vertices lie within the relaxed range of, say, $[\Delta/2,4\Delta]$, instead of $[\Delta,2\Delta]$. Our MPC algorithm is presented as Algorithm~\ref{alg:mpc_match} and takes as input a graph $G$ and parameter $k$, which specifies how many iterations of the main loop of \globalAlg{} it compresses to $O(\log k)$, using \RCAlg{}, i.e., the machinery developed in Section~\ref{sec:rc}. In order to efficiently apply \RCAlg{}, it has to sparsify the input graph. This happens in Lines~\ref{line:mpc_match:empty_multi}--\ref{line:mpc_match:add_edge}, in which it produces a multi-graph $G'$ that has maximum degree bounded by $k \cdot O(2^k\log n)$ with high probability. Since the goal is to simulate $k$ iterations of the main loop of \globalAlg{}, the algorithm independently samples $k$ subgraphs, which all become part of $G'$. Additionally, the sampling density is selected so that even when we consider the $i$-th iteration, $1 \le i \le k$, with the threshold $\Delta$ decreasing to $\Delta/2^i$, we are still likely to see enough edges to detect heavy vertices (i.e., the vertices that have degree approximately at least $\Delta/2^i$). The local distributed algorithm used for processing is \AlgFewPhases{} and is presented as Algorithm~\ref{alg:loc_peel}. It is easy to see that the iterations of its main loop correspond to the iterations of the main loop in \globalAlg{} with randomly downsampled--but still sufficiently dense---graph. In order to keep \AlgFewPhases{} deterministic, yet still allow for the selection of a random neighbor as a friend of a heavy vertex, we label all sampled edges with random integers from a large range in the process of creating $G'$ in Lines~\ref{line:mpc_match:randomize} and \ref{line:mpc_match:add_edge}. \AlgFewPhases{} can then select, for each heavy vertex $v$, the existing edge with the lowest such integer and make its other endpoint $v$'s friend $f(v)$. Note that \AlgFewPhases{} is in fact a deterministic local distributed algorithm and can therefore be accelerated, using \RCAlg{}. This follows from the fact that constant number of communication rounds between neighbors suffices for each peeling iteration in order to find out which neighbors are still present in the graph, select friends and, decide which vertices are being matched.

\begin{algorithm}[ht]
  \caption{$\AlgFewPhases_{(k,\lambda)}(G)$\algdesc{A local simulation of a few peeling iterations on a downsampled graph}}\label{alg:loc_peel}
  \KwIn{\\
    \qquad$\bullet$ multi-graph $G = (V,E)$ on $n$ vertices with edge labels $(i,\SET{\rho_u,\rho_w})$ for each edge $\SET{u,w}$ \\
    \qquad$\bullet$ a parameter $k > 1$ equal to the number of iterations to simulate\\
    \qquad$\bullet$ a scaling factor $\lambda$}
  \KwOut{}

\BlankLine

  $\Delta \leftarrow 2^k \lambda \log n$, $U \leftarrow V$, $M \leftarrow \emptyset$, $C \leftarrow \emptyset$
  
  \For{$i = 1\ldots k$}{

      $\Delta \leftarrow \Delta / 2$
      
      Let $G'$ be the subgraph of $G[U]$ consisting of edges labeled $(j,\SET{\rho_u,\rho_v})$ with $j=i$
      
      Let $H\subseteq U$ be the set of vertices of degree at least $\Delta$ in $G'$
      
      \ForEach{$v\in H$}{
      
          $N_v \leftarrow \mbox{set of neighbors of $v$ in $G'$}$
     
          $f(v) \leftarrow \mbox{$w \in N_v$ minimizing $\rho_v$ in the label $(i,\SET{\rho_v,\rho_w})$ associated with edge $(v,w)$}$
          
      }
      
      $F \leftarrow \{f(v):v\in H\}$
  
      Color each vertex in $H \cup F$ either blue or red independently at random\label{line:loc_peel:color}
      
      $M \leftarrow M \cup \left\{(v,f(v)) \in H \times F : \mbox{$v$ is blue} \land \mbox{$f(v)$ is red} \land \forall_{\mbox{\scriptsize blue\ }w \in H\setminus\{v\}}\mbox{$f(w) \ne f(v)$}\right\}$\label{line:loc_peel:match}
      
      $C \leftarrow C \cup H \cup F$\label{line:loc_peel:vc}
      
      $U \leftarrow U \setminus \left(H \cup F\right)$\label{line:loc_peel:reduce_graph}

  }

  \Return{$(M,C)$}
\end{algorithm}

\begin{algorithm}[ht]
  \caption{$\MPCMatch(G,k)$\algdesc{An MPC algorithm for maximum matching and vertex cover}}\label{alg:mpc_match}
  \KwIn{\\
    \qquad$\bullet$ graph $G = (V,E)$ on $n$ vertices\\
    \qquad$\bullet$ a number $k > 1$ of phases to execute at once}
  \KwOut{a matching and vertex cover}

\BlankLine

  $U \leftarrow v$, $\Delta \leftarrow n$, $M \leftarrow \emptyset$, $C \leftarrow \emptyset$
  
  $\lambda \leftarrow \mbox{sufficiently large constant}$
  
  \While{$\Delta > \lambda^2 \cdot \log n$}{
  
      \tcc{Invariant:\ maximum degree in $G[U]$ at most $2\Delta$ with high probability}
  
      Let $G$ be an empty edge--labeled multi-graph on $V$\label{line:mpc_match:empty_multi}
  
      $k' \leftarrow \MIN{k,\lceil\log\frac{\Delta}{\lambda^2\cdot\log n}\rceil}$
  
      \For{$i=1\ldots k'$}{
  
          \ForEach{\rm edge $(u,v)$ in $G[U]$}{
              
              Let $\rho_u$ and $\rho_v$ be independent random integers in $\left[0,1000n^{\Theta(1)}\right]$\label{line:mpc_match:randomize}
              
              With probability $\frac{2^{k'} \lambda \log n}{\Delta}$, add $(u,v)$ with label $(i,\SET{\rho_u,\rho_v})$ to $G'$\label{line:mpc_match:add_edge}
              
          }
      }
      
      $(M',C') \leftarrow \RCAlg(G',\AlgFewPhases_{(k',\lambda)},O(k'))$\label{line:mpc_match:rc}
      
      $U \leftarrow U \setminus C'$, $M \leftarrow M \cup M'$, $C \leftarrow C \cup C'$
      
      $\Delta \leftarrow \Delta / 2^{k'}$
  
  }

  Simulate $(M'',C'') \leftarrow \globalAlg(G[U],2\Delta)$ directly
  
  \Return{$(M\cup M'',C \cup C'')$}
\end{algorithm}

Note that we stop the accelerated simulation in \MPCMatch{} when the threshold $\Delta$ decreases to $\Theta(\log n)$. This follows from the fact that for small $\Delta$, the sampling may not result in sufficiently strong concentration in our relatively straightforward analysis. Hence at this point, we instead just directly simulate \globalAlg{} as stated by Fact~\ref{fact:sim_global_peel}.

We now state our main theorem and prove statements missing from the above description to obtain a formal proof.

\begin{theorem}\label{thm:MM_VC}
Let $\delta \in (0,1)$ be a fixed constant, $n$ be the number of vertices in the input graph, and $k \in [2,\log n]$. There is an MPC algorithm that with constant probability computes a $O(1)$--factor approximation to both maximum matching and minimum vertex cover, runs in $O\left(\frac{\log n}{k} \cdot \log k\right)$ rounds, and uses $\MAX{O(n^\delta),(2^k\log n)^{O(k)}}$ space per machine and $O(m + n\cdot(2^k\log n)^{O(k)})$ total space.
\end{theorem}

\begin{proof}
Note first that each peeling iteration is always conducted using an independently downsampled graph (with a different $i$ in edge labels $(i,\SET{\rho_u,\rho_v})$) and therefore, the approximate simulation of consecutive peeling iterations of \globalAlg{} is always conducted using ``fresh'' randomness, which allows for applying the Chernoff bounds combined with the union bound. Note also that if the constant $\lambda$ in Algorithm~\ref{alg:mpc_match} is sufficiently high, since all the thresholds we consider at least $(\lambda/2) \cdot \log n$, we know due to the Chernoff bound combined with the union bound that the probability that we err by a factor of more than $2$ in any of our degree estimates is bounded by $n^{-\Omega(1)}$, where the constant hidden by the big--Omega notation can be made arbitrarily large. This implies that the invariant stated in the comment in \MPCMatch{} holds with high probability throughout the execution of the algorithm. 

Also implied by the Chernoff and union bounds is the bound on the maximum degree in $G'$ created in Lines~\ref{line:mpc_match:empty_multi}--\ref{line:mpc_match:add_edge} of \MPCMatch{}, assuming that the main invariant holds. Again, with probability $1 - n^{-\Omega(1)}$, the maximum degree is bounded by $k \cdot 2 \cdot \frac{2^k \lambda \log n}{\Delta} \cdot 2\Delta = 2^{O(k)}\log n$. 
Applying Lemma~\ref{lem:rc}, this means that space per machine required for round compression in Line~\ref{line:mpc_match:rc} is $\MAX{O(n^\delta),(2^k\log n)^{O(k)}}$ and the total space is $O(m + n \cdot (2^k\log n)^{O(k)})$. All the other steps in \MPCMatch{} can routinely be executed in $O(1)$ MPC rounds, $O(n^\delta)$ space per machine and $O(nk + m)$ total space, which leads to the desired space bounds for the entire algorithm. Additionally, the total number of MPC rounds that the full execution of the algorithm uses is $O(\frac{\log n}{k} \log k) + O(\log \log n) = O(\frac{\log n}{k} \log k)$.

It remains to prove that the algorithm succeeds with constant probability at computing a constant factor approximation for both vertex cover and maximum matching. Let $C_i$ and $M_i$ be the vertices and edges added to the vertex cover and matching, respectively, in the simulation of the $i$-th iteration, where $i$ ranges from $1$ to $\log n + O(1)$. 
A given iteration may happen in either the accelerated execution of \AlgFewPhases{} or a direct simulation of \globalAlg{} at the end of \MPCMatch{}. As long as all the mentioned degree estimates are correct, which may not happen with probability at most $n^{-\Omega(1)}$ for iterations simulated using \AlgFewPhases, we claim that for each $i$, $E[|M_i|] = \Omega(E[|C_i|])$. 
To prove this, it suffices to notice that $|C_i| \le 2|H_i|$, where $H_i$ is the set of heavy vertices in the $i$-th iteration, and each given heavy vertex $v$ is matched with probability at least $\Omega(1)$ to its friend $f(v)$. First, the probability that $v$ and $f(v)$ are correctly colored---blue and red---is exactly $1/4$. Let us now bound the probability that another heavy vertex $w$ selects $f(v)$ to be its friend. For a given heavy neighbor $w$ of $f(v)$ this probability is at most $1/\Delta$, because $w$ has at least $\Delta$ neighbors in $G'$. The number of neighbors that $v$ has is then at most $4\Delta$, hence the probability of not being selected by any of them is at least $(1-1/\Delta)^{4\Delta} = \Omega(1)$, which combined with the probability of a good coloring is $\Omega(1)$ as well. By the linearity of expectation, we then get the desired property for a given $i$. Additionally, when the concentration bounds fail, which happens with miniscule probability of at most $n^{-\Omega(1)}$ with arbitrarily large constant in the exponent, we may lose at most $n/2$ in $E[|M_i|]$ and hence, in general we have $E[|M_i|] = \Omega(E[|C_i|]) - n^{-\Omega(1)}$. 

By summing over all iterations, $E[\sum_i|M_i|] = \Omega(E[\sum_i|C_i|]) - n^{-\Omega(1)}$. As long as the graph is non-empty (for empty graphs our algorithms outputs two empty sets, which are the optimal solution), this becomes $E[\sum_i|M_i|] = \Omega(E[\sum_i|C_i|])$, since significantly subconstant $n^{-\Omega(1)}$ is dwarfed by the positive size of any vertex cover. Since $\bigcup_i C_i$ and $\bigcup_i M_i$ are always a correct vertex cover and matching, respectively, $\sum_i| C_i| \ge \sum_i|M_i|$, because at least one endpoint of each edge in any matching has to be selected for a proper vertex cover. Hence with $\Omega(1)$ probability, $\sum_i|M_i| = \Omega(\sum_i|C_i|)$ (i.e., the inequality holds not only for the expectations, but also for actual sizes). Whenever this is the case $\bigcup_i M_i$ is a constant--factor approximation to maximum matching, because $\sum_i|C_i|$ is at least the maximum matching size. Similarly, $\bigcup_i M_i$ is a constant--factor approximation to vertex cover, because it is greater than some matching size by at most a constant factor.
\end{proof}

By setting $k = \sqrt{\log n/\log \log n}$, we obtain the following corollary.

\begin{corollary}[Restatement of Theorem~\ref{thm:MM_simple}]Let $\delta \in (0,1)$ be a fixed constant. There is an MPC algorithm that for an $n$--vertex graph, computes a constant--factor approximation to both maximum matching and minimum vertex cover, runs in $\tilde O(\sqrt{\log n})$ rounds, and uses $O(n^\delta)$ space per machine and $O(m) + n^{1+o(1)}$ total space.
\end{corollary}

\subsection{Extensions}\label{sec:mm_extensions}

Our algorithm can be adapted to obtain better guarantees and applied to the weighted matching problem. We now discuss possible extensions and adaptations one by one.

\begin{description}
  \item[Succeeding with higher probability:] To make the probability of success at least $1-\eps$ for $\eps > 0$, it suffices to run $O(\log(1/\eps))$ parallel instances of the algorithm for a given problem and then select the best of the solutions.

\item[$(2+\eps)$--approximation for maximum matching and vertex cover:] In order to compute a $(2+\eps)$--approximation for maximum matching, it suffices to repeat the procedure $O(\log(1/\eps))$ times. Each time we remove the discovered matching and continue with the subgraph induced by the unmatched vertices. Every time our algorithm succeeds, it decreases the matching size in the remaining graph by a constant factor until it becomes negligibly small. When this happens, we can ignore it and still obtain an almost $2$-approximation, since each edge in our matching can ``block'' at most two edges in the optimal solution.

For vertex cover, it suffices to output all endpoints in the matching that we have discovered above and a constant--factor approximation for the reminder. This stems from the relationship between matchings and vertex cover. For any matching, the vertex cover has to be at least the matching size, so the selection of endpoints of all matched edges introduces an approximation of factor at most 2. The residual graph has relatively small maximum matching and getting a constant factor approximation in it, leads to a $(2+\eps)$--approximation overall.

\item[$(1+\eps)$--Approximation for Maximum Matching:] To obtain this approximation, one can use a technique of McGregor~\cite{McGregor05} to find long augmenting paths, using algorithms that give a $(2+\eps)$--approximation. This leads to a multiplicative $(1/\eps)^{O(1/\eps)}$ increase in the number of rounds.

\item[Maximum weight matching:] In the \emph{maximum weight matching} problem, each edge has a positive weight and the goal is to discover a matching of maximum total weight of all involved edges.
A $(2+\eps)$--approximation can be obtained using a very simple approach of Lotker, Patt-Shamir, and Pettie \cite{Lotker:2015}, which uses $O(\log(1/\eps))$ iterations of any constant--factor approximation algorithm for the weighted case. This approach requires therefore constructing an $O(1)$--approximation algorithm for the weighted case. This can be done by considering $O(\log n)$ top classes of edge weights, $(2^{-(i+1)}w,2^{-i}w]$, where $w$ is the maximum weight of an edge and $i$ ranges between $0$ and $O(\log n)$. For edges belonging to each of these weight classes, we find a constant--factor approximation $M_i$ to (unweighted) maximum matching. It is easy to prove that the sum of weights of edges in all $M_i$'s is at least a constant times the weight of the optimal maximum weight matching. Unfortunately, edges in different $M_i$'s can share endpoints and $\bigcup M_i$ may not be a matching. To address this, one can prove that at most another constant factor is lost when one greedily adds edges when possible to an initially empty matching, considering weight classes from heaviest to lightest (i.e., in order $M_0$, $M_1$, etc.). This requires at most $O(\log n)$ parallel rounds and can be seen as a distributed algorithm in the LOCAL model on a graph of maximum degree $O(\log n)$. To turn in this into an MPC algorithm that requires less than a logarithmic number of rounds, one can use our round compression technique.
\end{description}

\section{Maximal Independent Set in Bounded--Arboricity Graphs}

In this section, we show a simple algorithm for finding a maximal independent set in graphs of bounded arboricity. 

\begin{theorem}\label{thm:arbMIS}
Let $\delta \in (0,1)$ be a fixed constant and let $\alpha \in [2,n]$ be an bound on the arboricity of the input graph. Let $n$ be the number of vertices in the input graph, and let $\gamma \in [2,n]$ be such that $\gamma / \alpha = 2^{2^{\Omega(\sqrt{\log\log n})}}$. 

There is an MPC algorithm that for an $n$--vertex graphs of arboricity bounded from above by $\alpha$, computes a maximal independent set with probability $1-O\left(\frac{\log n}{n}\right)$, runs in $O\left(\frac{\log n}{\log \gamma} \cdot \log \log \alpha\gamma\right)$ rounds, and uses $\MAX{O(n^\delta),2^{O(\log^2 \alpha\gamma)}}$ space per machine and $n \cdot 2^{O(\log^2 \alpha\gamma)}$ total space.
\end{theorem}

\begin{proof}
We present our algorithm as Algorithm~\ref{alg:arbMIS}. It repeatedly selects a set of low--degree vertices and simulates the efficient MIS algorithm of Ghaffari~\cite{G-MIS} on it via the round compression technique introduced in Section~\ref{sec:rc}. The discovered maximal independent set is added to the solution---variable $I$ in our pseudocode---and both vertices already in $I$ and their neighbors are removed from further consideration. The correctness of the algorithm---i.e., that it returns a maximal independent set---is straightforward assuming that the algorithm of Ghaffari does not fail to compute a corresponding MIS, which may happen with probability at most $1/n$ in each iteration. As we see in the next paragraph, the number of iterations is bounded by $O(\log n)$, and therefore, the algorithm does not err with probability greater than $O\left(\frac{\log n}{n}\right)$.

\begin{algorithm}[t]
  \caption{$\arbMIS(G,\alpha,\gamma)$\algdesc{An algorithm for finding MIS in a graph of bounded arboricity}}\label{alg:arbMIS}
  \KwIn{\\
    \qquad$\bullet$ graph $G = (V,E)$ on $n$ vertices\\
    \qquad$\bullet$ an upper bound $\alpha \in [2,n]$ on the arboricity of $G$\\
    \qquad$\bullet$ a progress factor $\gamma \in [2,n]$ such that $\gamma/\alpha = 2^{2^{\Omega(\sqrt{\log\log n})}}$}
  \KwOut{maximal independent set in $G$}

\BlankLine

  $U \leftarrow V$, $I \leftarrow \emptyset$, $\Delta \leftarrow 2\alpha\gamma$

  \While{$U \ne \emptyset$\label{line:arbMIS:loop}}{
  
      Let $U' \subseteq U$ be the set of vertices of degree at most $\Delta$ in $G[U]$\label{line:arbMIS:select_low}
      
      Let $G'$ be $G[U']$ with each vertex labeled with $O(\log^2 \Delta)$ random independent bits\label{line:arbMIS:random}
      
      Run $\RCAlg\left(G',\calI,O(\log \Delta) \right)$ where $\calI$ is the MIS algorithm of Ghaffari~\cite{G-MIS} adjusted to use the extra labels as the source of randomness\label{line:arbMIS:sim_g}
      
      Add the computed MIS for $G'$ to $I$\label{line:arbMIS:addToMIS}
      
      Remove from $U$ both $U'$ and all neighbors of vertices in $I$\label{line:arbMIS:purge_vertices}

  }
  
  \Return{$I$}
\end{algorithm}

We now analyze how many iterations of the loop in Line~\ref{line:arbMIS:loop} are necessary. Let $\Delta \eqdef 2\alpha\gamma$, defined as in the algorithm. Consider any graph of arboricity at most $\alpha$. The average vertex degree is bounded from above by $2\alpha$. Given a $\gamma \ge 2$, the fraction of vertices of degree greater than $\Delta$ can at most be $1/\gamma$. Since induced subgraphs of $G$ inherit the bound $\alpha$ on arboricity, $U$ and $U'$ are such that $|U'| \ge (1-1/\gamma) |U|$ after Line~\ref{line:arbMIS:select_low}. Therefore, the size of $U$ in Line~\ref{line:arbMIS:purge_vertices} decreases by a factor of at least $\gamma$. Hence the algorithm terminates after at most $\lceil\log_\gamma n\rceil = O(\frac{\log n}{\log \gamma})$ iterations of the loop in Line~\ref{line:arbMIS:loop}.

In order to simulate Ghaffari's algorithm in Line~\ref{line:arbMIS:sim_g}, we have to bound its number of rounds, the amount of randomness it needs, and the amount of local space at each vertex. 
First, his algorithm runs in $O(\log \Delta) + 2^{O(\sqrt{\log\log n})}$ rounds, which due to our requirement on $\gamma$, becomes $O(\log \Delta)$. By analyzing his algorithm, we learn that it runs in $O(\log \Delta)$ randomized rounds in which it needs a random bit that equals $1$ with probability $2^{-i}$, where $1 \le i \le O(\log \Delta)$. This random bit, independently of $i$ in this range, can be obtained from $O(\log \Delta)$ independent random bits distributed uniformly. Further rounds in his algorithm are purely deterministic. Hence, we need at most $O(\log^2 \Delta)$ random bits per vertex, which implies that each vertex is assigned a label of length at most $O(\log \Delta)$ words in Line~\ref{line:arbMIS:random}. Finally, during the execution of his algorithm no vertex requires more than a polynomial amount of space in the size of the neighborhood it can see (with the most non-obvious part being an application of the decomposition algorithm of Panconesi and Srinivasan~\cite{PanconesiS92}), hence the quantity $s_\calA$ in Lemma~\ref{lem:rc} is our case at most $\Delta^{O(\log \Delta)}$. Overall, the value of $s_\star$, as defined in Lemma~\ref{lem:rc}, is in our case $\Delta^{O(\log\Delta)} = 2^{O(\log^2 \alpha\gamma)}$. From Lemma~\ref{lem:rc}, we obtain the following features of the simulation of Ghaffari's algorithm in Line~\ref{line:arbMIS:sim_g}:
\begin{itemize}
 \item The number of MPC rounds is $O(\log \log \Delta) = O(\log\log \alpha\gamma)$.
 \item The amount of space per machine is $\MAX{2^{O(\log^2 \alpha\gamma)},O(n^\delta)}$.
 \item The total amount of space is $n \cdot 2^{O(\log^2 \alpha\gamma)}$.
\end{itemize}
The other steps in our algorithms (in particular, Lines~\ref{line:arbMIS:select_low}, \ref{line:arbMIS:random}, \ref{line:arbMIS:addToMIS}, and \ref{line:arbMIS:purge_vertices}) are easy to implement in $O(1)$ MPC rounds with $O(n^\delta)$ space per machine and $O(m + n \log\Delta) = n \cdot 2^{O(\log^2 \alpha\gamma)}$ total space, using sorting to distribute messages, where $m \le n \alpha$ is the number of edges in the input graph.
Combining this information with the bound on the number of iterations of the main loop in the algorithm, we obtain the desired MPC space and round bounds.
\end{proof}

By setting $\gamma = 2^{\sqrt{\log n / \log\log n}}$ in Theorem~\ref{thm:arbMIS}, we obtain the following corollary.

\begin{corollary}[Restatement of Theorem~\ref{thm:arbMIS_simple}]
Let $\calG$ be a class of graphs of arboricity $O(1)$ and let $\delta \in (0,1)$ be a fixed constant. There is an MPC algorithm that for an $n$--vertex graph from $\calG$, computes a maximal independent set with probability $1 - O\left(\frac{\log n}{n}\right)$, runs in $\tO(\sqrt{\log n})$ rounds, and uses $O(n^\delta)$ space per machine and $n^{1+o(1)}$ total space.
\end{corollary}

\bibliographystyle{alpha}
\bibliography{bibliography}

\end{document}